\documentclass[letterpaper, 10 pt, conference]{ieeeconf}  

\IEEEoverridecommandlockouts                              
\usepackage{graphics} 
\usepackage{epsfig} 
\usepackage{mathptmx} 
\usepackage{times} 
\usepackage{amsmath} 
\usepackage{amssymb}  
\usepackage{algorithm}
\usepackage{algpseudocode}
\usepackage{booktabs}
\usepackage{csvsimple}
\usepackage{booktabs}
\usepackage{datatool}
\usepackage{etoolbox}
\usepackage{tikz}
\usetikzlibrary{arrows.meta}

\usepackage{xcolor}

\definecolor{CORABlue}{RGB}{0,114,189}
\definecolor{CORARed}{RGB}{217,83,25}
\definecolor{CORAYellow}{RGB}{237,177,32}
\definecolor{CORAPurple}{RGB}{126,47,142}
\definecolor{CORAGreen}{RGB}{119,172,48}
\definecolor{CORALightBlue}{RGB}{77,190,238}
\definecolor{CORADarkRed}{RGB}{162,20,47}

\DTLloaddb{vdpEasy}{tables/VdP_easy_set.csv}

\newcommand{\RefForLine}[1]{%
\ifnumcomp{#1}{=}{2}{\cite{mu_set-based_2024}}{%
\ifnumcomp{#1}{=}{3}{\cite{scholte_nonlinear_2003}}{%
\ifnumcomp{#1}{=}{4}{\cite{wang_set-membership_2018}}{%
\ifnumcomp{#1}{=}{5}{\cite{combastel_zonotopes_2015}}{%
\ifnumcomp{#1}{=}{6}{\cite{bravo_bounded_2006}}{%
\ifnumcomp{#1}{=}{7}{\cite{scott_constrained_2016}}{%
\ifnumcomp{#1}{=}{8}{\cite{rego_guaranteed_2020}}{%
\ifnumcomp{#1}{=}{9}{\cite{khajenejad_guaranteed_2021}}{%
\ifnumcomp{#1}{=}{10}{\cite{khajenejad_guaranteed_2021}}{%
\ifnumcomp{#1}{=}{11}{--}{--}%
}}}}}}}}}}

\newcommand{\PrintCSVLineFrom}[2]{%
\hspace{0.3cm}\csvreader[
filter test=\ifnumcomp{\thecsvinputline}{=}{#2},
late after line=\\
]{#1}{1=\method,2=\time,3=\intvol,4=\width}{%
\method & \time & \intvol & \width & \RefForLine{#2}}%
}

\makeatletter
\let\NAT@parse\undefined
\makeatother

\usepackage[hidelinks]{hyperref}

\usepackage{eso-pic}
\usepackage{xcolor}

\newtheorem{theorem}{Theorem}
\newtheorem{remark}{Remark}
\newtheorem{proposition}{Proposition}

\title{\LARGE \bf
Divide and Discard: Fast Tightening of Guaranteed State Bounds for Nonlinear Systems 
}

\author{Nico Holzinger and Matthias Althoff
\thanks{This work was supported by the German Research Foundation (Deutsche Forschungsgemeinschaft, DFG) under grant GRK 2428.}
\thanks{N. Holzinger and M. Althoff are with Faculty of Computer Engineering,
        Technical University of Munich, 85748 Garching, Germany.
        {\tt\small nico.holzinger@tum.de; althoff@tum.de}}%
}

\begin{document}
\begin{filecontents*}{tables/VdP_easy_set.csv}
Method,Time_ms,Interval_volume,dirWidth
pDTDI,11.08,2.41,3.43
ESO-E,6.43,1.43,1.29
FRad-B,3.52,3.86,3.12
FRad-C,1.69,4.03,3.15
VolMin-B,2.05,4.27,5.47
CZN-A,15.50,2.69,2.31
CZMV,15.65,1.58,1.38
CZKH,3.52,2.41,2.11
ZBKH,$\mathbf{1.35}$,4.83,3.81
DD,1.64,$\mathbf{1.00}$,$\mathbf{1.00}$
Best (abs.),,0.25,0.40
\end{filecontents*}

\maketitle
\thispagestyle{empty}
\pagestyle{empty}

\begin{abstract}

We propose a simple yet effective divide-and-discard (DD) approach to guaranteed state estimation for nonlinear discrete-time systems. Our method iteratively subdivides interval enclosures of the state and propagates them forward in time using a mean-value enclosure. The central idea is to rely on repeated refinement of simple sets rather than on more complex set representations, yielding an observer that is straightforward to implement and easy to integrate into existing frameworks. Our divide-and-discard strategy exploits that many sets can be discarded early and limits the number of maintained sets, resulting in low computational cost with complexity that scales only quadratically in the state dimension. The proposed method is evaluated on nonlinear benchmark problems previously used to compare guaranteed observers, where it outperforms state-of-the-art approaches in terms of both computational efficiency and enclosure tightness.

\end{abstract}


\AddToShipoutPictureFG*{%
  \AtPageLowerLeft{%
    \ifnum\value{page}=1
      \makebox[\paperwidth]{%
        \hfill
        \raisebox{1.2cm}[0pt][0pt]{%
          \parbox{0.9\paperwidth}{\centering\footnotesize
          This work has been submitted to the IEEE for possible publication.
          Copyright may be transferred without notice, after which this version
          may no longer be accessible.}%
        }%
        \hfill
      }%
    \fi
  }%
}

\section{INTRODUCTION}

State estimation is a core component of control and monitoring systems, providing information about internal state variables that are not directly measurable. In many safety-critical applications, such as autonomous driving \cite{althoff_online_2014}, or surgical robots \cite{haidegger_autonomy_2019}, estimation errors may lead to constraint violations or unsafe behavior. While probabilistic estimators are widely used, they typically characterize uncertainty in terms of likelihood or confidence. In contrast, set-based state estimation computes state enclosures that are guaranteed to contain the true system state for all admissible disturbances and measurement errors, thereby enabling robust control \cite{schurmann_reachset_2018}, rigorous prediction of system behaviors \cite{althoff_set_2021}, and robust fault detection \cite{chabane_fault_2015,mu_set-based_2024}.

A broad range of guaranteed state estimation methods has been proposed for nonlinear systems, where the choice of set representation plays a decisive role. Simple representations, such as intervals \cite{jaulin_nonlinear_2002,yang_accurate_2018,dinh_unified_2024}, ellipsoids \cite{gollamudi_set-membership_1998,scholte_nonlinear_2003}, or parallelotopes \cite{chisci_recursive_1996,qu_nonlinear_2022} offer advantages in terms of computational efficiency, but often lead to overly conservative enclosures when applied directly to nonlinear dynamics. More expressive representations, such as zonotopes \cite{combastel_state_2005,alamo_guaranteed_2005,le_zonotopic_2013}, zonotope bundles \cite{khajenejad_guaranteed_2021}, or constrained zonotopes \cite{scott_constrained_2016,rego_guaranteed_2020}, typically provide tighter enclosures of the state, but at the cost of increased algorithmic complexity, nontrivial implementation effort, and additional tuning parameters. As a result, scalability and practical deployment can become challenging. While each class of methods has its merits, there is currently no approach that consistently outperforms all others. In particular, for highly nonlinear and higher-dimensional systems, many methods struggle to maintain tight enclosures, as demonstrated in \cite{holzinger_comparison_2026}.

We show that strong estimation performance does not necessarily require complex set representations by proposing a novel divide-and-discard (DD) observer combining simple intervals with systematic refinement. The key idea is to control conservatism at its source: for standard mean-value enclosures, outer-approximation grows primarily with the diameter of the propagated set, so refinement directly improves tightness by keeping local domains small. While subdivision-based strategies are, in principle, susceptible to exponential growth with dimension, the proposed observer is not affected because (i) measurements discard many subsets and (ii) the number of active subsets is explicitly capped.

The contribution of this paper is a guaranteed state observer for nonlinear discrete-time systems that combines reachability-based prediction, direct strip contraction, and persistent refinement in a novel divide-and-discard scheme. We develop a refinement strategy tailored to state estimation with a capped number of subsets, lightweight redundancy pruning, and immediate discarding of infeasible subsets. Our systematic evaluation shows that this combination yields consistently tight enclosures at favorable runtimes on the considered benchmarks.

This paper is organized as follows. Section~\ref{sec:preliminaries} introduces the problem formulation and required preliminaries, while Section~\ref{sec:method} presents the proposed observer. The evaluation on established nonlinear benchmarks is given in Section~\ref{sec:eval}, and concluding remarks are provided in Section~\ref{sec:conclusions}.


\section{PRELIMINARIES}
\label{sec:preliminaries}

This section introduces the problem setting and the basic set-based methods required to formalize the proposed observer. We first recall standard set operations used in
set-based state estimation, then define the estimation problem, and finally introduce range bounding techniques.

\subsection{Basic Set Operations}

Considering sets $\mathcal{X},\mathcal{W} \subset \mathbb{R}^n$, $\mathcal{Y} \subset \mathbb{R}^p$ and a real matrix $C \in \mathbb{R}^{p \times n}$, the basic set operations Minkowski sum, linear map and intersection are respectively defined as \cite[Def. 2]{scott_constrained_2016}
\begin{align}
\mathcal{X} \oplus \mathcal{W} &:= \{x + w \mid x \in \mathcal{X}, w \in \mathcal{W}\}, \\
 C \mathcal{X} &:= \{C x \mid x \in \mathcal{X}\}, \\
 \mathcal{X} \cap_C \mathcal{Y} &:= \{ x \mid x \in \mathcal{X},\; Cx \in \mathcal{Y} \}.
\end{align}    
These operations are used throughout the paper within prediction and correction
steps.

\subsection{Problem Statement}

We consider nonlinear discrete-time systems of the form
\begin{equation}
x_{k+1} = f(x_k,u_k,w_k),
\label{eq:system}
\end{equation}
with state $x_k \in \mathbb{R}^n$, known input $u_k \in \mathbb{R}^m$, and unknown but
bounded disturbance $w_k \in \mathcal{W} \subseteq \mathbb{R}^n$. Measurements are given by
\begin{equation}
y_k = C x_k + v_k,
\label{eq:measurement}
\end{equation}
where $C \in \mathbb{R}^{p \times n}$ and $v_k \in \mathcal{V} \subseteq \mathbb{R}^p$
denotes bounded measurement noise. The initial state satisfies
$x_0 \in \mathcal{X}_0$.

Starting from $\mathcal{X}_0$, the set of all states consistent with the system dynamics,
inputs, measurements, and bounded uncertainties at time step $k+1$ is defined recursively
as
\begin{equation}
\begin{aligned}
\mathcal{X}_{k+1}^\star =
\Big\{ &x_{k+1} \;\Big|\;
x_k \in \mathcal{X}_k^\star,\;
w_k \in \mathcal{W},\;
v_{k+1} \in \mathcal{V},\\
&x_{k+1} = f(x_k,u_k,w_k),\;
y_{k+1} = Cx_{k+1} + v_{k+1}
\Big\}.
\label{eq:true_state_set}
\end{aligned}
\end{equation}  

The objective of guaranteed state estimation is to compute, at each time step $k$, a set
$\mathcal{X}_k$ that outer-approximates $\mathcal{X}_k^\star$ while remaining as tight as
possible.

\begin{remark}
  We focus on systems with linear measurement equations in this work since many practical examples are covered by this structure and to keep the correction step simple. There are several approaches in the literature to handle nonlinear measurement equations, e.g., by outer-approximating the set of measurement-consistent states as in \cite[Property 2]{alamo_guaranteed_2005} or by using nonlinear contractors \cite[Sec. 4]{jaulin_applied_2001}. The proposed observer can be extended to nonlinear measurements by adapting the correction step accordingly.
\end{remark}

\subsection{Intervals and Range Bounding}
In this work, sets are represented using intervals. A scalar interval
$[\,\underline{x},\overline{x}\,] \subset \mathbb{R}$ is defined as
\begin{equation}
[\,\underline{x},\overline{x}\,] := \{ x \in \mathbb{R} \mid \underline{x} \leq x \leq \overline{x} \}.
\end{equation}
An interval in $\mathbb{R}^n$ is understood componentwise as an interval vector,
\begin{equation}
\begin{aligned}
\mathcal{X} =
\big[ [\underline{x}_1,\overline{x}_1], \dots, [\underline{x}_n,\overline{x}_n] \big]^\top ,
\qquad \mathcal{X} \in \mathbb{IR}^n ,
\end{aligned}
\end{equation}

where $\mathbb{IR}^n$ denotes the set of all interval vectors in $\mathbb{R}^n$. In the following, the term interval is also used for elements of $\mathbb{IR}^n$ whenever the dimension is clear from the context.

To compute the predicted set, range bounding of the nonlinear dynamics is required. Let $\mathcal{X}_k\subset\mathbb{R}^n$ be a set and $f\colon \mathbb{R}^n\rightarrow \mathbb{R}^n$ be a function, then $f(\mathcal{X}_k) = \left\{f(x)\ \middle|\ x \in \mathcal{X}_k\right\}$. The goal is to compute a set
$\mathcal{X}^p_{k+1}$ such that
\begin{equation}
f(\mathcal{X}_k) \subseteq \mathcal{X}^p_{k+1}.
\end{equation}
Several techniques for range bounding of nonlinear functions exist, including interval arithmetic and Taylor-model methods, reviewed in \cite{althoff_range_bounding_2018}. In this work, nonlinear propagation is performed using a first-order mean-value enclosure. In particular, for any continuously differentiable mapping $f$ and any interval $\mathcal{X} \in \mathbb{IR}^n$ \cite[Sec.~2.4.3]{jaulin_applied_2001},
\begin{equation}
f(\mathcal{X}) \subseteq f(c) + J(\mathcal{X})\,(\mathcal{X} - c),
\qquad c \in \mathcal{X},
\label{eq:mv}
\end{equation}
where $J(\mathcal{X})$ is an interval enclosure of the Jacobian of $f$ over $\mathcal{X}$.

While first-order enclosures may be conservative for strongly nonlinear systems, this effect
can be mitigated through systematic subdivision and repeated correction of interval sets,
as employed by our observer.

\section{DIVIDE-AND-DISCARD OBSERVER}
\label{sec:method}

This section presents the proposed divide-and-discard observer. We first introduce a single-set observer using a single interval and then extend it by a refinement strategy. Algorithm \ref{alg:IntBB_hull} summarizes the complete procedure.

\begin{algorithm}[t]
\caption{Divide-and-Discard Observer}
\label{alg:IntBB_hull}
\begin{algorithmic}[1]
\Require Initial set $\mathcal{X}_0$, inputs $\{u_k\}$, measurements $\{y_k\}$
\Ensure Guaranteed state enclosures $\{\mathcal{X}_k\}$

\State $\mathcal{X}_0 \gets \Call{GSContract}{\mathcal{X}_0,C,\mathcal{V},y_0,I_{\max}}$
\For{$k = 0,1,\dots,N-1$}
    \State Refine $\mathcal{X}_k$ until $M_{\max}$ subsets are active (Sec.\ref{sec:refine}) \label{alg:refine}
    \State $\mathcal{X}_{k+1}^p \gets f(\mathcal{X}_k,u_k,\mathcal{W})$ \Comment{predict} \label{alg:predict}
    \State $\mathcal{X}_{k+1} \gets \Call{GSContract}{\mathcal{X}_{k+1}^p,C,\mathcal{V},y_{k+1},I_{\max}}$
    \State $\mathcal{X}_{k+1} \gets \Call{Prune}{\mathcal{X}_{k+1}}$ \label{alg:prune}\Comment{remove intervals fully contained in others}
\EndFor

\Statex \rule{\linewidth}{0.4pt}

\Function{GSContract}{$\mathcal{X},C,\mathcal{V},y,I_{\max}$}
    \Comment{$\mathcal{X}$ is a collection of intervals $\mathcal{X}_b \in \mathbb{IR}^n$}
    \For{\textbf{each} interval $\mathcal{X}_b \in \mathcal{X}$} \label{alg:eachbox}
        \For{$t = 1,\dots,I_{\max}$} \Comment{Gauss-Seidel iterations} \label{GSSweeps}
            \For{$i=1,\dots,p$} \Comment{measurement strips}
                \For{$j=1,\dots,n$} \Comment{state variables}
                    \If{$C_{(i,j)} \neq 0$}
                        \State compute admissible interval $\mathcal{I}_j^{(i)}$
                        \Statex \hspace{\algorithmicindent} \hspace{3cm} for $x_j$ using \eqref{eq:admissible_interval}
                        \State $\mathcal{X}_b^{(j)} \gets \mathcal{X}_b^{(j)} \cap \mathcal{I}_j^{(i)}$ \label{alg:intersecting}
                    \EndIf
                \EndFor
            \EndFor
            \If{$\mathcal{X}_b$ is empty}
                \State discard $\mathcal{X}_b$ \label{alg:discard} and \textbf{break}
            \EndIf
            \If{$\textit{nothing changed}$} \textbf{break} \EndIf
        \EndFor
    \EndFor
    \State \Return remaining (tightened) intervals in $\mathcal{X}$
\EndFunction
\end{algorithmic}
\end{algorithm}

\subsection{Single-Set Observer}
The single-set observer follows a classical prediction-correction procedure and represents the state at time step $k$ by a single interval enclosure $\mathcal{X}_k \in \mathbb{IR}^n$. This observer provides a simple yet effective reference for assessing the benefits of additional refinement strategies.

Starting from the state enclosure $\mathcal{X}_{k}$, the prediction step performs a set-based evaluation of the dynamics \eqref{eq:system} to compute a set of states reachable at the next time step prior to incorporating new measurement information,
\begin{equation}
\begin{aligned}
\mathcal{X}^p_{k+1} = f(\mathcal{X}_{k},u_k,\mathcal{W}).
\end{aligned}
\end{equation}
Nonlinear propagation is carried out using the mean-value enclosure \eqref{eq:mv}, applied to the augmented variable $z = [x^\top \; w^\top]^\top$ over the joint domain $\mathcal{X}_k \times \mathcal{W}$, while treating the known input $u_k$ as a fixed parameter, similar to \cite[Thm.~2]{rego_guaranteed_2020}.

Measurement information is incorporated by tightening the predicted enclosure $\mathcal{X}^p_{k+1}$ with measurement-consistent constraints. For linear measurement equations of the form \eqref{eq:measurement}, each measurement component induces a strip
constraint in the state space. Writing $y_{k+1}=[y_{k+1}^{(1)},\ldots,y_{k+1}^{(p)}]^\top$
and denoting by $C_{(i,.)}\in\mathbb{R}^{1\times n}$ the $i$-th row of $C$, the $i$-th strip is \cite[Property 2]{alamo_guaranteed_2005}
\begin{equation}
\begin{aligned}
\mathcal{S}_{k+1}^{(i)} := \left\{x\in\mathbb{R}^n \,\middle|\, y_{k+1}^{(i)}-C_{(i,.)}x \in [\underline v_i,\overline v_i]\right\},\quad i=1,\ldots,p.
\end{aligned}
\end{equation}

Ideally, the correction step would compute the exact intersection of $\mathcal{X}^p_{k+1}$ with these strips. Since this intersection is generally not representable as an interval, we compute an interval outer-approximation via a Gauss-Seidel (GS) tightening procedure (Alg.~\ref{alg:IntBB_hull}, Function GSContract). The GS tightening acts as an interval contractor for the linear strip inequalities \cite[Chapter~4]{jaulin_applied_2001} and yields a tight interval enclosure of the exact intersection. Related ideas appear in observer designs based on constraint propagation, where the whole estimation problem is formulated as a constraint propagation problem as in \cite{gning_constraints_2006}; here, in contrast, GS is used specifically as a lightweight contractor within the correction step.

Let $\mathcal{X}^p_{k+1} = [\underline{x},\overline{x}] \in \mathbb{IR}^n$ denote the
predicted interval enclosure and consider a measurement strip induced by the $i$-th
measurement component,
\begin{equation}
  \label{eq:stripConstraint}
y_{k+1}^{(i)} - \overline{v}_i \;\le\; C_{(i,.)} x \;\le\; y_{k+1}^{(i)} - \underline{v}_i .
\end{equation}

The enclosure is tightened componentwise. Fix a state variable $x_j$ with
$C_{(i,j)} \neq 0$ and decompose the linear expression as
\begin{equation}
  \label{eq:sDef}
C_{(i,.)} x = C_{(i,j)} x_j + s,
\qquad
s := \sum_{\ell \neq j} C_{(i,\ell)} x_\ell .
\end{equation}
Using the current bounds of $\mathcal{X}^p_{k+1}$, the contribution of the remaining
state variables is enclosed by interval arithmetic as
\begin{equation}
  \label{eq:sEnclosure}
s \in \sum_{\ell \neq j} C_{(i,\ell)} [\,\underline{x}_\ell,\overline{x}_\ell\,]
= [\,\underline{s},\overline{s}\,],
\end{equation}
where $\underline{s}$ and $\overline{s}$ are obtained by evaluating the linear expression
over the interval bounds.

Substituting \eqref{eq:sDef} and \eqref{eq:sEnclosure} into the strip constraint \eqref{eq:stripConstraint} yields
\begin{align}
\underbrace{y_{k+1}^{(i)} - \overline{v}_i - \overline{s}}_{=: \,\underline b}
\le
C_{(i,j)} x_j
\le
\underbrace{y_{k+1}^{(i)} - \underline{v}_i - \underline{s}}_{=: \,\overline b}.
\end{align}
Solving the inequality for $x_j$ amounts to dividing by $C_{(i,j)}$. Since $\operatorname{sign}(C_{(i,j)})$ is a priori unknown, we use the $\min$ and $\max$ operators to obtain the correct ordering of the bounds for both $C_{(i,j)}>0$ and $C_{(i,j)}<0$. Entries with $C_{(i,j)}=0$ are skipped in the construction of $\mathcal{I}_j^{(i)}$. For $ C_{(i,j)}\neq 0$ this leads to the interval
\begin{equation}
\label{eq:admissible_interval}%
\scalebox{1.1}{$
\mathcal{I}_j^{(i)}=
\left[
\min\!\left(\frac{\underline b}{C_{(i,j)}},\frac{\overline b}{C_{(i,j)}}\right),\;
\max\!\left(\frac{\underline b}{C_{(i,j)}},\frac{\overline b}{C_{(i,j)}}\right)
\right].
$}
\end{equation}

The GS contractor tightens bounds by intersecting the $j$-th state variable interval of $\mathcal{X}^p_{k+1}$ with $\mathcal{I}_j^{(i)}$ and repeats this sequentially over strips and components for up to $I_{\max}$ Gauss-Seidel iterations (Alg. \ref{alg:IntBB_hull}, Line \ref{GSSweeps}), terminating early if an iteration yields no further tightening. This implements a sound strip contractor and yields a measurement-consistent interval correction. The proposed SSO is straightforward to implement and computationally efficient. However, since the tightness of interval-based propagation degrades as set diameters grow under nonlinear dynamics, we introduce a refinement strategy in the following subsection.

\subsection{Refinement Strategy}
\label{sec:refine}
Instead of enclosing the state by a single interval, the proposed approach maintains a finite collection of intervals whose union outer-approximates the set of possible states,
\begin{equation}
\mathcal{X}_k = \bigcup_{b=1}^{M_k} \mathcal{X}_k^{(b)}, \qquad
\mathcal{X}_k^{(b)} \in \mathbb{IR}^n .
\label{eq:union_intervals}
\end{equation}

Refinement is performed by subdividing intervals while respecting the cap $M_{\max}$ (Alg.~\ref{alg:IntBB_hull}, Line~\ref{alg:refine}). Let the current collection be $\{\mathcal{X}^{(b)}\}_{b=1}^{M}$ with $\mathcal{X}^{(b)}=[\underline{x}^{(b)},\overline{x}^{(b)}]$. We introduce a componentwise scaling vector $s\in\mathbb{R}^n_{>0}$ and measure widths as
\begin{equation}
  \label{eq:scaled_width}
\tilde w_{b,j} := \frac{\overline{x}^{(b)}_{j}-\underline{x}^{(b)}_{j}}{s_j}, \qquad b=1,\dots,M,\ j=1,\dots,n.
\end{equation}

For each interval, we define
\begin{equation}
  \label{eq:sidelength}
j_b^\star = \arg\max_{j\in\{1,\dots,n\}} \tilde w_{b,j}.
\end{equation}
As long as $2M\leq M_{\max}$, every active interval is bisected once along its widest scaled dimension $j_b^\star$, so that the number of subsets doubles. The scaling of the width avoids a bias towards state variables with larger physical units. If $M < M_{\max} < 2M$, exactly $M_{\max}-M$ intervals are split. To select them in linear time, the values $\tilde w_{b,j_b^\star}$ are not sorted but grouped into $K_{\mathrm{split}}$ equal-width bins spanning their full range in the current collection, which yields a coarse ranking from small to large widths, similar in spirit to successive-binning selection methods for avoiding full sorting \cite{tibshirani2008median}; see Fig.~\ref{fig:binning}. Starting from the largest-width bin, intervals are selected until the split budget is exhausted, taking intervals within each bin in encounter order. Each selected interval $\mathcal{X}^{(b)}$ is bisected at the midpoint of its $j_b^\star$-th dimension. This preserves the preference for refining wide intervals while keeping the splitting complexity linear in $M_{\max}$.

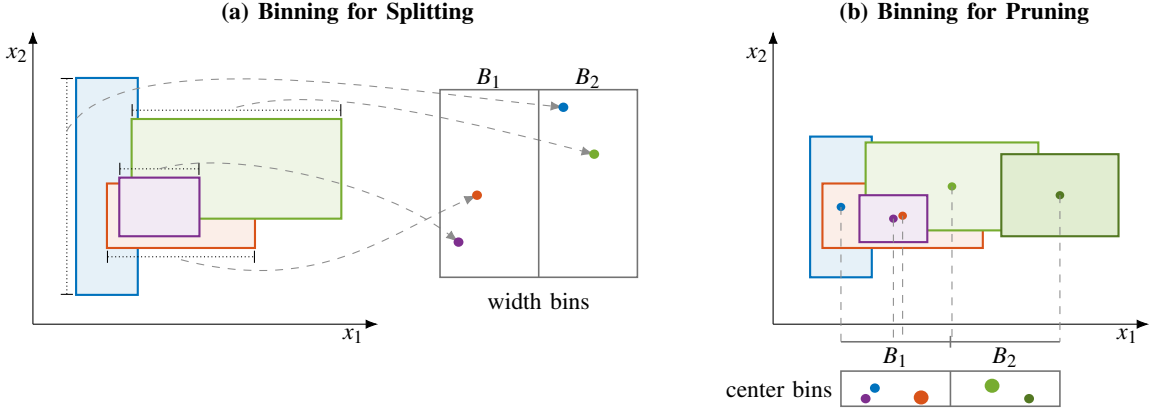
\begin{figure*}[t]
\centering
\begin{tikzpicture}[x=0.82cm,y=0.78cm,>=Latex,
boxA/.style={draw=CORABlue, thick, fill=CORABlue!10},
boxB/.style={draw=CORARed, thick, fill=CORARed!10},
boxC/.style={draw=CORAGreen, thick, fill=CORAGreen!12},
boxD/.style={draw=CORAPurple, thick, fill=CORAPurple!10},
boxE/.style={draw=CORAGreen!70!black, thick, fill=CORAGreen!22},
binline/.style={draw=black!55, line width=0.6pt},
proj/.style={draw=black!45, dashed, line width=0.4pt},
meas/.style={draw=black!55, dashed, line width=0.6pt},
lab/.style={font=\small},
every node/.style={font=\small}
]

\node[font=\small\bfseries] at (5.5,5.8) {(a) Binning for Splitting};

\draw[->] (0.4,0.5) -- (0.4,5.5);
\draw[->] (0.4,0.5) -- (6.0,0.5);
\node[lab] at (5.6,0.3) {$x_1$};
\node[lab] at (0.15,5.1) {$x_2$};

\draw[boxA] (1.1,1.0) rectangle (2.1,4.7);
\draw[boxB] (1.6,1.8) rectangle (4.0,2.9);
\draw[boxC] (2.0,2.3) rectangle (5.4,4.0);
\draw[boxD] (1.8,2.0) rectangle (3.1,3.0);


\draw[|-|,densely dotted] (0.95,1.0) -- (0.95,4.7);

\draw[|-|,densely dotted] (1.6,1.65) -- (4.0,1.65);

\draw[|-|,densely dotted] (2.0,4.15) -- (5.4,4.15);

\draw[|-|,densely dotted] (1.8,3.15) -- (3.1,3.15);

\draw[proj,->] (0.95,3.8)  .. controls (1.8,5.3) and (6.2,4.5) .. (9.0,4.2); 
\draw[proj,->] (2.8,1.6)  .. controls (5.0,1.0) and (6.1,2.0) .. (7.6,2.7); 
\draw[proj,->] (3.7,4.2)  .. controls (4.8,4.5) and (6.2,4.3) .. (9.5,3.4); 
\draw[proj,->] (2.45,3.15) .. controls (3.6,3.5) and (6.1,2.9) .. (7.3,1.9); 

\draw[binline] (7.0,1.3) rectangle (10.2,4.5);
\draw[binline] (8.6,1.3) -- (8.6,4.5);

\node[lab] at (8.6,0.9) {width bins};
\node[lab] at (7.8,4.7) {$B_1$};
\node[lab] at (9.4,4.7) {$B_2$};

\fill[CORAPurple] (7.3,1.9) circle (0.08);
\fill[CORARed] (7.6,2.7) circle (0.08);
\fill[CORABlue] (9.0,4.2) circle (0.08);
\fill[CORAGreen] (9.5,3.4) circle (0.08);

\node[font=\small\bfseries] at (15.5,5.8) {(b) Binning for Pruning};

\draw[->] (12.4,0.5) -- (12.4,5.5);
\draw[->] (12.4,0.5) -- (18.5,0.5);
\node[lab] at (18.15,0.3) {$x_1$};
\node[lab] at (12.15,5.1) {$x_2$};

\draw[boxA] (13.0,1.3) rectangle (14.0,3.7);
\draw[boxB] (13.2,1.8) rectangle (15.8,2.9);
\draw[boxC] (13.9,2.1) rectangle (16.7,3.6);
\draw[boxE] (16.1,2.0) rectangle (18.0,3.4);
\draw[boxD] (13.8,1.9) rectangle (14.9,2.7);

\fill[CORABlue]             (13.5,2.5) circle (0.07);
\fill[CORARed]              (14.5,2.35) circle (0.07);
\fill[CORAGreen]            (15.3,2.85) circle (0.07);
\fill[CORAGreen!70!black]   (17.05,2.7) circle (0.07);
\fill[CORAPurple]           (14.35,2.3) circle (0.07);

\draw[binline] (13.5,0.2) -- (17.05,0.2);
\draw[binline] (15.275,0.1) -- (15.275,0.3);

\draw[proj] (13.5,2.5) -- (13.5,0.2);
\draw[proj] (14.5,2.35) -- (14.5,0.2);
\draw[proj] (15.3,2.85) -- (15.3,0.2);
\draw[proj] (17.05,2.7) -- (17.05,0.2);
\draw[proj] (14.35,2.3) -- (14.35,0.2);

\draw[binline] (13.5,-0.90) rectangle (17.05,-0.30);
\draw[binline] (15.275,-0.90) -- (15.275,-0.30);

\node[lab] at (12.5,-0.6) {center bins};
\node[lab] at (14.3875,-0.07) {$B_1$};
\node[lab] at (16.1625,-0.07) {$B_2$};

\fill[CORABlue]         (14.05,-0.59) circle (0.08);
\fill[CORARed] (14.8,-0.75) circle (0.12);
\fill[CORAPurple] (13.9,-0.77) circle (0.08);
\fill[CORAGreen]  (15.95,-0.55) circle (0.12);
\fill[CORAGreen!70!black]  (16.55,-0.77) circle (0.08);

\end{tikzpicture}
\caption{Running example of the binning heuristics used for splitting and pruning. In (a), four two-dimensional intervals are binned by their largest scaled side lengths $\tilde w_{b,j_b^\star}$, shown as dotted lines. Here, $M_{\max}=5$ is chosen, so one interval from the large-width bin $B_2$ is selected for splitting; in this case, the green interval. The resulting five intervals are then propagated and corrected by the single-set observer; subfigure (b) illustrates how the intervals are binned again according to their centers along the state component with the largest spread according to \eqref{eq:centerSpread}. From each bin, only the interval with the largest scaled side length is selected as candidate container, indicated by the larger points in the bins. Since the purple interval is contained in the orange interval, it is pruned.}
\label{fig:binning}
\end{figure*}

All intervals are propagated using the SSO. Subsets that become empty after applying the GS contractor are discarded (Alg.~\ref{alg:IntBB_hull}, Line~\ref{alg:discard}). To reduce redundancy, a partial containment check is performed, and intervals found to be contained in others are removed. Specifically, for two intervals
$\mathcal{X}^{(b)}=[\underline{x}^{(b)},\overline{x}^{(b)}]$ and
$\mathcal{X}^{(c)}=[\underline{x}^{(c)},\overline{x}^{(c)}]$, we remove $\mathcal{X}^{(c)}$ if there exists an index $b\neq c$ such that
$\underline{x}^{(b)}\le \underline{x}^{(c)}$ and $\overline{x}^{(c)}\le \overline{x}^{(b)}$, i.e., $\mathcal{X}^{(c)}\subseteq \mathcal{X}^{(b)}$. Let $c_b$ denote the center of $\mathcal{X}^{(b)}$, and let
\begin{equation}
  \label{eq:centerSpread}
j_{\mathrm{ctr}}^\star
=
\arg\max_{j\in\{1,\dots,n\}}
\Big(
\max_{b=1,\dots,M} c_{b,j}-\min_{b=1,\dots,M} c_{b,j}
\Big)
\end{equation}
be the state component with the largest spread of centers. To keep computational costs low, the scalar center values $c_{b,j_{\mathrm{ctr}}^\star}$ are partitioned into $K_{\mathrm{prune}}$ equal-width bins spanning their full range, analogously to the splitting step; see Fig.~\ref{fig:binning}. Within each bin, only the interval with the largest scaled side length $\tilde w_{b,j_b^\star}$ is retained as candidate container, and every other interval is checked only against this representative. After pruning, the remaining subsets are kept without being merged.

This refinement mechanism is related to subdivision-based methods such as subpaving and SIVIA \cite{kieffer_guaranteed_1998,jaulin_set_1993}. However, its role here is different. Rather than recursively classifying boxes in a set-inversion framework, the intervals are propagated forward in time, contracted by the measurement update, and discarded if they become empty. In addition, overlaps between subsets are permitted, which avoids the repartitioning typically required in classical subpaving constructions \cite[Ch.~3]{jaulin_applied_2001}.

\begin{theorem}[Soundness]
\label{thm:soundness}
Let $\mathcal{X}_0$ be an initial enclosure with $x_0\in \mathcal{X}_0$, and assume $w_k\in \mathcal{W}$ and $v_k\in \mathcal{V}$ for all $k$. The true state is contained in
\[
\mathcal{X}_k=\bigcup_{b=1}^{M_k} \mathcal{X}_k^{(b)},\qquad \mathcal{X}_k^{(b)}\in \mathbb{IR}^n.
\]
\end{theorem}

\begin{proof}
We show that each step in Sec.~\ref{sec:method} preserves inclusion.
\emph{Splitting:} Each $\mathcal{X}_k^{(b)}$ is bisected into two children whose union equals the parent, hence no feasible state is removed.
\emph{Prediction:} At time step $k$, the input $u_k$ is fixed, while the state and disturbance vary over $\mathcal{X}_k^{(b)}$ and $\mathcal{W}$, respectively. Applying \eqref{eq:mv} to the dynamics over $\mathcal{X}_k^{(b)} \times \mathcal{W}$ yields a predicted interval containing all states $f(x,u_k,w)$ with $x \in \mathcal{X}_k^{(b)}$ and $w \in \mathcal{W}$ \cite[Sec.~2.4.3]{jaulin_applied_2001}.
\emph{Correction and discarding:} The Gauss-Seidel tightening implements a sound strip contractor for the linear inequalities \cite[Thm.~4.2]{jaulin_applied_2001}. Therefore, contraction does not remove measurement-consistent states. If a contracted subset becomes empty, it contains no feasible state and can be discarded.
\emph{Pruning:} A subset is removed only if it is contained in another retained subset, which leaves the union $\mathcal{X}_k$ unchanged.
Starting from $x_0\in \mathcal{X}_0$, the above implications yield $x_k\in \mathcal{X}_k$ for all $k$ by induction.
\end{proof}

\subsection{Computational Complexity}
The computational complexity of one observer step is summarized in the following proposition.
\begin{proposition}[Per-step worst-case complexity]
\label{prop:complexity}
The worst-case complexity per time step of Alg. \ref{alg:IntBB_hull} is
\[
\mathcal{O}\!\big(M_{\max}pn^2\big).
\]
\end{proposition}
\vspace{0.3cm}
\begin{proof}
We bound the runtime by summing the worst-case costs of splitting, prediction, correction, and pruning. \emph{Splitting:} In Line~\ref{alg:refine}, splitting one interval requires $\mathcal{O}(n)$ operations. Since the number of active intervals is capped by $M_{\max}$ and the selection rule in the final partial round is linear, refinement requires $\mathcal{O}(nM_{\max})$ binary operations. \emph{Prediction and Correction:} In Line~\ref{alg:predict}, computing the mean-value enclosure for one interval requires $\mathcal{O}(n^2)$ operations, due to the Jacobian evaluation and the subsequent interval update. In \textsc{GSContract}, for all combinations of $p$ measurement strips and $n$ state components, an update involving an inner sum over $n$ components is performed, resulting in $\mathcal{O}(pn^2)$ per GS iteration. Since the number of GS iterations is bounded by the fixed constant $I_{\max}$, the correction cost per interval is $\mathcal{O}(pn^2)$. As at most $M_{\max}$ intervals are processed after Line~\ref{alg:refine}, the total prediction and correction costs are $\mathcal{O}(M_{\max}n^2)$ and $\mathcal{O}(M_{\max}pn^2)$, respectively. \emph{Pruning:} In Line~\ref{alg:prune}, intervals are partitioned into a fixed number of groups, and each interval is checked only against one representative. Since grouping is linear and each containment check costs $\mathcal{O}(n)$, pruning requires $\mathcal{O}(nM_{\max})$. Summing the bounds and discarding lower-order terms yields an overall complexity of $\mathcal{O}(M_{\max}pn^2)$ binary operations, proving the claim.
\end{proof}

If $M_{\max}$ is kept constant, the per-step complexity reduces to $\mathcal{O}(pn^2)$. This compares favorably to many established classes of guaranteed observers, such as zonotope- and constrained-zonotope-based approaches, whose reported complexity expressions often involve cubic or higher-order terms; see, for example, \cite[Tab.~1]{rego_guaranteed_2020}.

\section{Experimental Results}
\label{sec:eval}
This section evaluates the proposed observer and analyzes its performance across different benchmark settings.
\subsection{Evaluation Setup}
Our observer is evaluated on two nonlinear benchmarks previously used for comparing guaranteed state estimators in \cite{holzinger_comparison_2026}. The first is a discrete-time Van der Pol oscillator, representing a low-dimensional nonlinear system with tunable nonlinearity parameter $\mu$. The second is a nonlinear multi-tank system, providing a higher-dimensional scalability benchmark. Details on the benchmark dynamics, measurement models, and uncertainty sets are given in Appendix \ref{app:benchmarks}. For all observers, parameters are tuned using Bayesian optimization as described in Appendix~\ref{app:parameter}. All algorithms are implemented in CORA \cite{althoff_introduction_2015} using MATLAB and will be made publicly available to ensure repeatability. Simulations are run on an AMD Ryzen 7 with 3.8 GHz and 64 GB RAM.

We assess runtime (average time per step) and enclosure tightness. Ideally, tightness would be measured by the exact set volume, but this is not uniformly tractable across all set representations considered. We therefore report two complementary surrogates. First, the interval-hull volume,
\begin{equation}
\tilde v \;=\; \frac{1}{n_{\mathrm{steps}}}\sum_{k=1}^{n_{\mathrm{steps}}} \big(\mathrm{vol}(\mathrm{hull}(\mathcal{X}_k))\big)^{1/n},
\end{equation}
where $\mathrm{hull}(\cdot)$ denotes the interval hull operator. This metric is deterministic, but it may disadvantage richer set representations. Second, a representation-independent mean-width measure,
\begin{equation}
\tilde w \;=\; \frac{1}{N\,n_{\mathrm{steps}}}\sum_{k=1}^{n_{\mathrm{steps}}}\sum_{i=1}^{N}\Big(\rho(\mathcal{X}_k,d_i)+\rho(\mathcal{X}_k,-d_i)\Big),
\end{equation}
with $N=10n$ random unit directions $d_i$ and support function $\rho(\mathcal{X},d)\;=\;\max_{x\in \mathcal{X}} d^\top x$. This metric better reflects geometry across methods, but it is approximated from sampled directions. To ease direct comparison, both $\tilde v$ and $\tilde w$ are normalized by the minimum value achieved across all compared observers. The best achieved absolute values of both metrics are also reported providing further insights.

\subsection{Main Results}
The proposed observer is compared against the best-performing methods from \cite{holzinger_comparison_2026} on the two benchmark settings that form the main comparison in this work: the highly nonlinear Van der Pol oscillator with $\mu=5$ and the 30-tank system. Table~\ref{tab:combined_results} reports enclosure tightness and computation time; additional results are given in the ablation study. Methods that fail due to set explosion are listed separately below the tables.

\begin{table*}[t]
  \centering
  \caption{Comparison of observer performance on the highly nonlinear Van der Pol oscillator and the 30-tank system.}
  \label{tab:combined_results}
  \renewcommand{\arraystretch}{1.0}
  \setlength{\tabcolsep}{12pt}
  \begin{tabular}{lccccccc}
    \toprule
    & \multicolumn{3}{c}{Van der Pol oscillator} & \multicolumn{3}{c}{30-tank system} & \\
    \cmidrule(lr){2-4} \cmidrule(lr){5-7}
    Method & Time [ms] & $\hat{v}$ & $\hat{w}$ & Time [ms] & $\hat{v}$ & $\hat{w}$ & Ref. \\
    \midrule
    \multicolumn{8}{l}{\textit{Intervals}}\\
    pDTDI    & 15.00 & 3.61 & 7.34 & 71.88 & 1.47 & 2.32 & \RefForLine{2} \\
    \multicolumn{8}{l}{\textit{Zonotopes}}\\
    FRad-B   & 2.89 & 5.25 & 6.54 & 13.57 & 6.33 & 2.76 & \RefForLine{4} \\
    FRad-C   & -- & -- & -- & 11.63 & 6.05 & 2.38 & \RefForLine{5} \\
    VolMin-B & 2.49 & 3.50 & 5.25 & 48.82 & 4.51 & 1.93 & \RefForLine{6} \\
    \multicolumn{8}{l}{\textit{Constrained zonotopes}}\\
    CZN-A    & 26.88 & 3.71 & 5.44 & 603.59 & 6.36 & 2.01 & \RefForLine{7} \\
    CZMV     & 37.60 & 4.63 & 7.53 & 22197.41 & 2.34 & 1.25 & \RefForLine{8} \\
    CZKH     & 3.61 & 3.03 & 3.92 & 396.17 & 3.87 & 2.08 & \RefForLine{9} \\
    \multicolumn{8}{l}{\textit{Zonotope bundles}}\\
    ZBKH     & 6.76 & 15.35 & 42.49 & 64.33 & 3.95 & 1.82 & \RefForLine{10} \\
    \multicolumn{8}{l}{\textit{Ours}}\\
    DD       & $\mathbf{1.24}$ & $\mathbf{1.00}$ & $\mathbf{1.00}$ & $\mathbf{7.65}$ & $\mathbf{1.00}$ & $\mathbf{1.00}$ & \RefForLine{12} \\
    Best (abs.) &  & 0.25 & 0.53 &  & 0.24 & 5.67 & \\
    \addlinespace
    \multicolumn{8}{l}{Set explosion occurred for ESO-E \cite{scholte_nonlinear_2003} on both benchmarks and for FRad-C on the Van der Pol benchmark.}\\
    \bottomrule
  \end{tabular}
\end{table*}

As shown in Table~\ref{tab:combined_results}, the proposed DD observer achieves the tightest enclosures and the lowest computation time on both benchmarks. For the highly nonlinear Van der Pol oscillator, the gain in tightness is particularly pronounced, indicating that repeated refinement effectively reduces the conservatism caused by strong nonlinearity. For the 30-tank system, the DD observer again attains the tightest enclosure while remaining the fastest method, showing that propagating multiple subsets does not compromise scalability in the higher-dimensional setting. This is consistent with the low cost of interval arithmetic and the efficient parallel treatment of intervals. Bayesian optimization selected $M_{\max}=251$ for the Van der Pol oscillator and $M_{\max}=246$ for the 30-tank system.

\subsection{Ablation Study}
To further assess the behavior of the proposed DD observer, we examine the effects of weaker nonlinearity, increased uncertainty, and the parameter $M_{\max}$.

\paragraph{Van der Pol with $\mu = 0.1$}
We first revisit the Van der Pol benchmark in a weakly nonlinear regime by setting $\mu=0.1$. The results in Table~\ref{tab:vdp_easy_results} show that the proposed observer already improves substantially over the best-performing method from \cite{holzinger_comparison_2026} in terms of enclosure tightness, while retaining highly competitive computation times, shown here for $M_{\max} = 213$. Thus, the proposed refinement strategy is beneficial not only in strongly nonlinear settings, but already in a comparatively mild regime.

\setlength{\tabcolsep}{11pt}
\begin{table}[ht]
  \centering
  \caption{Results for the Van der Pol oscillator with $\mu = 0.1$.}
  \label{tab:vdp_easy_results}
  \renewcommand{\arraystretch}{1}
  \begin{tabular}{lcccc}
    \toprule
    Method & Time [ms] & $\hat{v}$ & $\hat{w}$ & Ref. \\
    \midrule
    \multicolumn{5}{l}{\textit{Intervals}}\\
    \PrintCSVLineFrom{tables/VdP_easy_set.csv}{2}
    \multicolumn{5}{l}{\textit{Ellipsoids}}\\
    \PrintCSVLineFrom{tables/VdP_easy_set.csv}{3}
    \multicolumn{5}{l}{\textit{Zonotopes}}\\
    \PrintCSVLineFrom{tables/VdP_easy_set.csv}{4}
    \PrintCSVLineFrom{tables/VdP_easy_set.csv}{5}
    \PrintCSVLineFrom{tables/VdP_easy_set.csv}{6}
    \multicolumn{5}{l}{\textit{Constr.\ zonotopes}}\\
    \PrintCSVLineFrom{tables/VdP_easy_set.csv}{7}
    \PrintCSVLineFrom{tables/VdP_easy_set.csv}{8}
    \PrintCSVLineFrom{tables/VdP_easy_set.csv}{9}
    \multicolumn{5}{l}{\textit{Zonotope bundles}}\\
    \PrintCSVLineFrom{tables/VdP_easy_set.csv}{10}
    \multicolumn{5}{l}{\textit{Ours}}\\
    \PrintCSVLineFrom{tables/VdP_easy_set.csv}{11}
    \csvreader[
      filter test=\ifnumcomp{\thecsvinputline}{=}{12}
    ]{tables/VdP_easy_set.csv}{1=\method,2=\time,3=\intvol,4=\width}{%
      \method &  & \intvol & \width & %
    }\\
    \bottomrule
  \end{tabular}
\end{table}

\paragraph{Increased Uncertainty}
To isolate the effect of stronger uncertainty, we additionally consider both benchmarks with tenfold process disturbances and fivefold measurement noise. For the Van der Pol oscillator with $\mu=5$, Table~\ref{tab:vdp_hard_high_results} shows that many competing methods fail to return usable enclosures over the full simulation horizon, whereas the DD observer remains stable. In this setting, Bayesian optimization selects the larger value $M_{\max}=659$, indicating that under strong nonlinearity and increased uncertainty, maintaining more subsets remains beneficial for controlling overestimation.

 \setlength{\tabcolsep}{8pt}
\begin{table}[ht]
  \centering
  \caption{Results for the Van der Pol oscillator with $\mu = 5$ and increased uncertainty.}
  \label{tab:vdp_hard_high_results}
  \renewcommand{\arraystretch}{1}
  \begin{tabular}{lcccc}
    \toprule
    Method & Time [ms] & $\hat{v}$ & $\hat{w}$ & Ref. \\
    \midrule
    \multicolumn{5}{l}{\textit{Intervals}}\\
    \PrintCSVLineFrom{tables/VdP_hard_high_set.csv}{2}
    \multicolumn{5}{l}{\textit{Zonotopes}}\\
    \PrintCSVLineFrom{tables/VdP_hard_high_set.csv}{4}
    \PrintCSVLineFrom{tables/VdP_hard_high_set.csv}{6}
    \multicolumn{5}{l}{\textit{Constr.\ zonotopes}}\\
    \PrintCSVLineFrom{tables/VdP_hard_high_set.csv}{7}
    \PrintCSVLineFrom{tables/VdP_hard_high_set.csv}{8}
    \multicolumn{5}{l}{\textit{Ours}}\\
    \PrintCSVLineFrom{tables/VdP_hard_high_set.csv}{11}
    \begin{filecontents*}{tables/VdP_hard_high_set.csv}
Method,Time_ms,Interval_volume,dirWidth
pDTDI,39.91,3.46,6.14
ESO-E,NaN,Inf,Inf
FRad-B,2.64,455.47,62\,718.09
FRad-C,NaN,Inf,Inf
VolMin-B,$\mathbf{2.50}$,40.29,311.78
CZN-A,31.10,42.36,353.57
CZMV,14.38,$4.11\times 10^{7}$,$2.91 \times 10^{13}$
CZKH,NaN,Inf,Inf
ZBKH,NaN,Inf,Inf
DD,3.17,$\mathbf{1.00}$,$\mathbf{1.00}$
Best (abs.),,0.97,1.84
\end{filecontents*}
\csvreader[
      filter test=\ifnumcomp{\thecsvinputline}{=}{12}
    ]{tables/VdP_hard_high_set.csv}{1=\method,2=\time,3=\intvol,4=\width}{%
      \method &  & \intvol & \width & %
    }\\
    \addlinespace
    \multicolumn{5}{l}{Set explosion: ESO-E, FRad-C, CZKH, ZBKH.}\\
    \bottomrule
  \end{tabular}
\end{table}

For the 30-tank system with increased uncertainty, Table~\ref{tab:tank30_high_results} shows that enclosure tightness deteriorates for all methods, and the competing approaches come closer to the DD observer. Here, the benefit of additional subsets is smaller than in the nominal setting, since the larger uncertainty reduces the chances of discarding subsets and thus limits the gains from more aggressive refinement. Accordingly, the tuning procedure selects the much smaller value $M_{\max}=3$. Nevertheless, the DD observer still attains the tightest enclosure while maintaining a clear advantage in computation time.

\setlength{\tabcolsep}{11pt}
\begin{table}[ht]
  \centering
  \caption{Results for the 30-tank benchmark with increased uncertainty.}
  \label{tab:tank30_high_results}
  \renewcommand{\arraystretch}{1}
  \begin{tabular}{lcccc}
    \toprule
    Method & Time [ms] & $\hat{v}$ & $\hat{w}$ & Ref. \\
    \midrule
    \multicolumn{5}{l}{\textit{Intervals}}\\
    \PrintCSVLineFrom{tables/Tank30_high.csv}{2}
    \multicolumn{5}{l}{\textit{Constr.\ zonotopes}}\\
    \PrintCSVLineFrom{tables/Tank30_high.csv}{8}
    \PrintCSVLineFrom{tables/Tank30_high.csv}{9}
    \multicolumn{5}{l}{\textit{Zonotope bundles}}\\
    \PrintCSVLineFrom{tables/Tank30_high.csv}{10}
    \multicolumn{5}{l}{\textit{Ours}}\\
    \PrintCSVLineFrom{tables/Tank30_high.csv}{11}
    \begin{filecontents*}{tables/Tank30_high.csv}
Method,Time_ms,IntervalVolume,Rel_Width
pDTDI,70.87,1.15,1.37
ESO-E,NaN,Inf,Inf
FRad-B,NaN,Inf,Inf
FRad-C,NaN,Inf,Inf
VolMin-B,NaN,Inf,Inf
CZN-A,NaN,Inf,Inf
CZMV,2408.14,4.01,1.33
CZKH,531.97,3.71,1.60
ZBKH,55.81,3.14,1.43
DD,$\mathbf{1.39}$,$\mathbf{1.00}$,$\mathbf{1.00}$
Best (abs.),,0.98,11.97
\end{filecontents*}
\csvreader[
      filter test=\ifnumcomp{\thecsvinputline}{=}{12}
    ]{tables/Tank30_high.csv}{1=\method,2=\time,3=\intvol,4=\width}{%
      \method &  & \intvol & \width & %
    }\\
    \addlinespace
    \multicolumn{5}{l}{Set explosion: ESO-E, FRad-B, FRad-C, VolMin-B, CZN-A.}\\
    \bottomrule
  \end{tabular}
\end{table}

\paragraph{Parameter Sweep}
To study the role of $M_{\max}$ more systematically, we sweep this parameter for the highly nonlinear Van der Pol benchmark and the nominal 30-tank system. Since $M_{\max}=1$ corresponds to the underlying single-set observer, this directly quantifies the benefit of the divide-and-discard strategy. Figures~\ref{fig:VdP_sweep} and \ref{fig:Tank30_sweep} report average computation time per step and mean width; in each case, the best value achieved by the competing methods is shown as a reference.

For the Van der Pol oscillator, Fig.~\ref{fig:VdP_sweep} confirms the approximately linear growth of computation time with $M_{\max}$ predicted by Prop.~\ref{prop:complexity}. At the same time, tightness improves rapidly for small values of $M_{\max}$: the DD observer already outperforms the best competing method at $M_{\max}=3$, while gains beyond roughly $100$ are marginal. Hence, a relatively small number of subsets already captures most of the benefit of refinement in this nonlinear benchmark.

\begin{figure}
  \centering
  \includegraphics[width=0.45\textwidth]{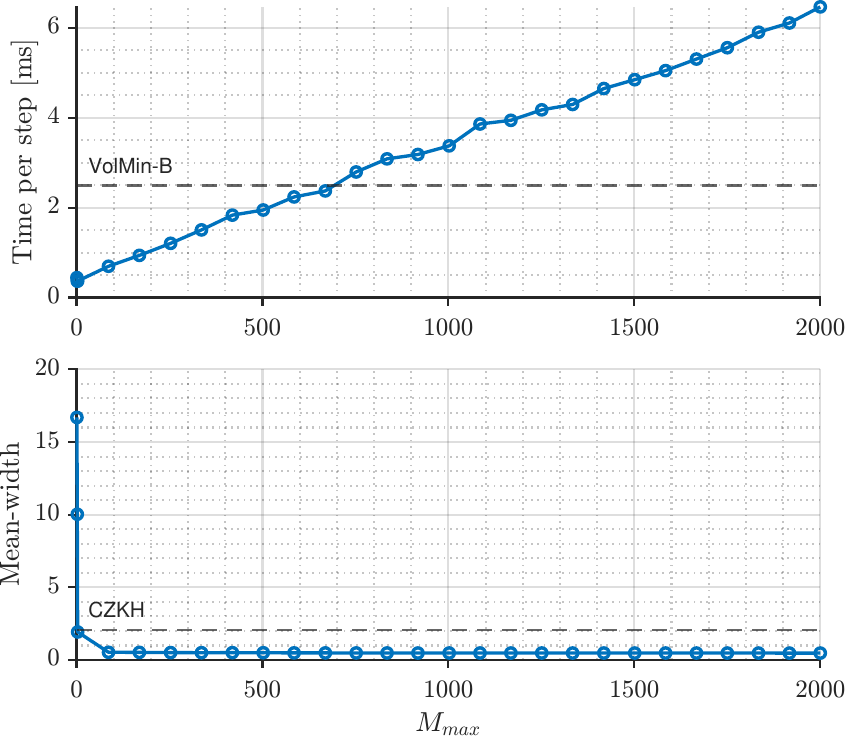}
  \caption{Effect of maximum number of intervals $M_{\max}$ on tightness and computation time for the Van der Pol oscillator with $\mu = 5$.}
  \label{fig:VdP_sweep}
\end{figure}

For the 30-tank system, Fig.~\ref{fig:Tank30_sweep} again shows a nearly linear increase in computation time with $M_{\max}$, while the improvement in tightness is more gradual. This is consistent with the weaker nonlinearity and more favorable measurement structure of this benchmark, for which even the single-set observer already performs strongly relative to most competitors. Still, increasing $M_{\max}$ reduces conservatism further, and the DD observer surpasses the best competing method once $M_{\max}$ is greater than $100$. This sweep suggests that the proposed observer is comparatively robust to the choice of $M_{\max}$, since strong performance is achieved over a broad range of values. This makes parameter tuning straightforward, as it mainly serves to refine the trade-off between tightness and computation time rather than to recover usable performance.

\begin{figure}
  \centering
  \includegraphics[width=0.45\textwidth]{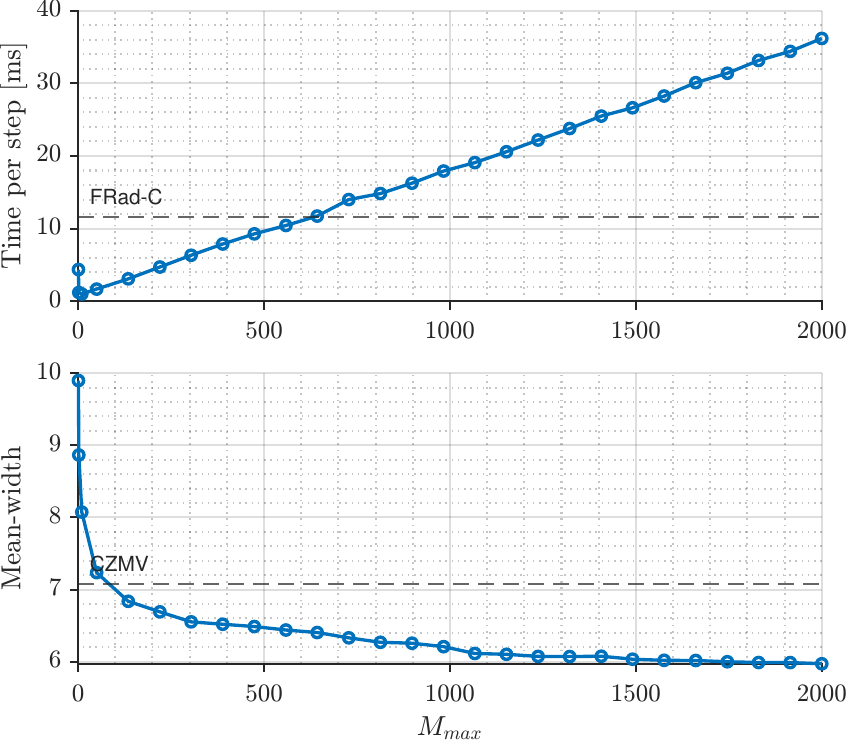}
  \caption{Effect of maximum number of intervals $M_{\max}$ on tightness and computation time for the 30-tank benchmark.}
  \label{fig:Tank30_sweep}
\end{figure}

\section{CONCLUSIONS}
\label{sec:conclusions}

This paper presented a novel divide-and-discard observer for guaranteed state estimation of nonlinear discrete-time systems. Our approach combines simple interval-based prediction and intersection-based correction with a systematic refinement strategy that reduces conservatism through subdivision of the state space.

The numerical evaluation shows that the proposed observer consistently achieves the tightest state enclosures across the considered comparisons and nonlinear systems. In particular, for strongly nonlinear dynamics, the refinement strategy significantly reduces conservatism relative to existing methods, yielding substantially improved estimation accuracy; this advantage persists in higher-dimensional settings, although the relative improvement is smaller for systems with weaker nonlinearities. A key reason for the observed computational efficiency is that the underlying interval operations are inherently cheap, even in higher dimensions; moreover, maintaining multiple subsets per step does not substantially increase cost as long as the number of active intervals remains moderate, since the computations can be carried out in a vectorized manner.

 \addtolength{\textheight}{-1cm}   

\section{ACKNOWLEDGMENTS}

The authors gratefully acknowledge the financial support from the research training group ConVeY funded by the German Research Foundation under grant GRK 2428.

\appendix
\subsection{Benchmark Details}
\label{app:benchmarks}

This appendix summarizes the benchmark models and uncertainty settings used for the evaluation.

\paragraph{Van der Pol Oscillator}
The first benchmark is an Euler-discretized Van der Pol oscillator with sampling time $h$,
\begin{equation}
\begin{aligned}
x_{1,k+1} &= x_{1,k} + h\,x_{2,k} + w_{1,k}, \\
x_{2,k+1} &= x_{2,k} + h\big(\mu(1-x_{1,k}^2)x_{2,k} - x_{1,k}\big) + w_{2,k},
\end{aligned}
\end{equation}
where $\mu>0$ is the nonlinearity parameter. The measurement equation is
\begin{equation}
y_k = Cx_k + v_k, \qquad C = [\,1\ \ 0\,], \qquad v_k \in \mathcal{V},
\end{equation}
so only the first state component is measured. In the experiments, we use $h=0.025$. The initial and uncertainty sets are
\begin{equation}
\mathcal{X}_0=\mathcal{B}^2,\qquad
\mathcal{W}=10^{-3}\mathcal{B}^2,\qquad
\mathcal{V}=0.2\,\mathcal{B},
\end{equation}
where $\mathcal{B}^n$ denotes the unit box in $\mathbb{R}^n$. By varying $\mu$, this benchmark is used to assess how the observer handles increasing nonlinearity.

\paragraph{Multi-Tank System}
The second benchmark is a nonlinear multi-tank system with state $x_k\in\mathbb{R}^n$, whose dynamics follow Torricelli's law. Using explicit Euler discretization with sampling time $h$ gives
\begin{align}
x_{1,k+1} = x_{1,k} + h&\left(-\kappa_1\sqrt{2g}\sqrt{x_{1,k}} + (Bu_k)_1\right) + w_{1,k}, \\
x_{j,k+1} = x_{j,k} + h&\left(\kappa_j\sqrt{2g}\big(\sqrt{x_{j-1,k}}-\sqrt{x_{j,k}}\big) + (Bu_k)_j\right)\\
 &+ w_{j,k}, \quad j=2,\dots,n, \nonumber
\end{align}
with measurement model
\begin{equation}
y_k = Cx_k + v_k,\qquad v_k\in\mathcal{V}.
\end{equation}
The benchmark is used to assess scalability with increasing state dimension. In the experiments, we use $h=0.5$, $g=9.81\,\mathrm{m/s^2}$, and $\kappa_j=0.015$ for all $j$. The initial and uncertainty sets are
\begin{equation}
\mathcal{X}_0 = 20\,\mathbf{1}_n \oplus 4\,\mathcal{B}^n,\qquad
\mathcal{W} = 10^{-3}\,\mathcal{B}^n,\qquad
\mathcal{V} = 0.2\,\mathcal{B}^p,
\end{equation}
where $\mathbf{1}_n$ is the all-ones vector. The tank indices with external inflow and available measurements for the 30-tank system are listed in Table~\ref{tank_indices}.

\begin{table}[t]
\centering
\caption{Tank indices with external inflow and measurements for the 30-tank system.}
\label{tank_indices}
\renewcommand{\arraystretch}{1}
\setlength{\tabcolsep}{10pt}
\scriptsize
\begin{tabular}{@{}ll@{}}
\toprule
\textbf{Type} & \textbf{Tank indices}\\
\midrule
Water inflow & 1, 4, 5, 7, 9, 10, 13, 15, 16, 19, 21, 22, 25, 27, 28\\[1mm]
Water level measurement & 2, 4, 5, 7, 8, 10, 11, 13, 14, 16, 17, 19, 20,\\
& 21, 22, 23, 25, 26, 27, 28, 29\\
\bottomrule
\end{tabular}
\normalsize
\vspace{-1mm}
\end{table}

\subsection{Parameter Tuning}
\label{app:parameter}
Parameters for all observers are tuned using Bayesian optimization with a fixed budget to minimize a weighted combination of computation time and enclosure tightness. For the proposed DD observer, the main tuning parameter is the maximum number of intervals $M_{\max}$, varied in the range $1$ to $2000$. The influence of $M_{\max}$ is examined separately in the corresponding study. Since the maximum number of Gauss--Seidel iterations and the numbers of bins used for splitting and pruning were found to have comparatively little influence, they are fixed throughout all experiments to
\begin{equation}
I_{\max}=5,\qquad K_{\mathrm{split}}=K_{\mathrm{prune}}=20.
\end{equation}
For the competing methods, zonotope orders, the maximum number of constraints for constrained-zonotope observers, the H-set size for ZBKH and CZKH, and the number of splits for pDTDI are tuned analogously.

\bibliographystyle{IEEEtran}
\begin{small}
  \bibliography{references}
\end{small}

\end{document}